%% file: main.tex
\documentclass[letterpaper,11pt]{article}
\usepackage{fullpage,amsfonts,amssymb,amsmath,amsthm,hyperref,latexsym,graphicx,color,cancel,xstring,mleftright,authblk}

\newif\ifqcrypt
\qcryptfalse
\usepackage[utf8]{inputenc} % allow utf-8 input
\usepackage[T1]{fontenc}    % use 8-bit T1 fonts
\usepackage{hyperref}       % hyperlinks
\usepackage{url}            % simple URL typesetting
\usepackage{booktabs}       % professional-quality tables
\usepackage{amsfonts}       % blackboard math symbols
\usepackage{nicefrac}       % compact symbols for 1/2, etc.
\usepackage{microtype}      % microtypography
\usepackage{xcolor}         % colors
\numberwithin{equation}{section}
\usepackage{amsmath}
\usepackage{cryptocode}

\usepackage[capitalize]{cleveref}

\DeclarePairedDelimiter\angles{\langle}{\rangle}
\newcommand{\ot}{\otimes}

\renewcommand{\epsilon}{\varepsilon}

\newcommand{\ketbra}[2]{\lvert#1\rangle\langle #2\rvert }

\usepackage{comment}

\usepackage{mathtools}
\usepackage{enumerate}

\usepackage{amsfonts}
\usepackage{amssymb}
\usepackage{amsmath}

\newcommand{\negl}{\mathsf{negl}}

\usetikzlibrary{quantikz}

\newcommand{\N}{\mathbb{N}}

\newcommand{\SGen}{\mathsf{SGen}}
\newcommand{\Sign}{\mathsf{Sign}}
\newcommand{\Ver}{\mathsf{Ver}}
\newcommand{\vk}{\mathsf{vk}}

  \theoremstyle{definition}
\newtheorem{theorem}{Theorem}[section]
\newtheorem{claim}[theorem]{Claim}
\newtheorem{definition}[theorem]{Definition}

\newtheorem{lemma}[theorem]{Lemma}

\newcommand{\adv}{\mathcal{A}}

\newlength\myindent
\input{macros.tex}
\begin{document}
\title{Robust Quantum Public-Key Encryption \\ with Applications to Quantum Key Distribution}
\date{}
\author[1,2]{Giulio Malavolta}
\author[3]{Michael Walter}
\affil[1]{Bocconi University, Milan, Italy}
\affil[2]{Max Planck Institute for Security and Privacy, Bochum, Germany}
\affil[3]{Ruhr-Universit\"{a}t Bochum, Bochum, Germany}

\maketitle
\ifqcrypt
\else
\begin{abstract}
Quantum key distribution (QKD) allows Alice and Bob to agree on a shared secret key, while communicating over a public (untrusted) quantum channel.
Compared to classical key exchange, it has two main advantages:
(i)~The key is \emph{unconditionally} hidden to the eyes of any attacker, and
(ii)~its security assumes only the existence of authenticated classical channels which, in practice, can be realized using Minicrypt assumptions, such as the existence of digital signatures.
On the flip side, QKD protocols typically require multiple rounds of interactions, whereas classical key exchange can be realized with the minimal amount of two messages using public-key encryption.
A long-standing open question is whether QKD requires more rounds of interaction than classical key exchange.

In this work, we propose a two-message QKD protocol that satisfies \emph{everlasting} security, assuming only the existence of quantum-secure one-way functions. That is, the shared key is unconditionally hidden, provided computational assumptions hold during the protocol execution. Our result follows from a new construction of quantum public-key encryption (QPKE) whose security, much like its classical counterpart, only relies on authenticated \emph{classical} channels.
\end{abstract}
\fi

\ifqcrypt
\else
%=============================================================================
\section{Introduction}
%=============================================================================
\fi
Quantum key distribution (QKD)~\cite{BB84} enables Alice and Bob to exchange a secret key over a public (untrusted) quantum channel.
Compared to classical key exchange, it offers two main advantages:
(i)~It hides the key \emph{unconditionally} or \emph{information-theoretically} to the eyes of any (possibly unbounded and quantum) attacker, and (ii)~it relies only on the existence of authenticated classical channels, which in practice can be instantiated using Minicrypt~\cite{Impagliazzo} computational assumptions.
The former is plainly impossible to achieve without quantum information, and we also have strong evidence~\cite{DBLP:conf/stoc/ImpagliazzoR89} that classical key exchange requires more structured (Cryptomania) computational assumptions.

Over the past decades, QKD has inspired a staggering amount of research, ranging from profound theoretical works~\cite{MY98,LC99,Mayers01,Renner} all the way to large-scale experiments~\cite{jouguet2013experimental,korzh2015provably,yin2017satellite}.
As such, QKD is one of the most studied topics in the theory quantum information.
Despite the vast literature on the topic, QKD protocols still do not outperform classical key exchange in all aspects:
Whereas classical key exchange can be realized using two messages~\cite{DH76}, which is optimal, to the best of our knowledge all known QKD protocols require more than two rounds of interaction.

This raises the question of whether more rounds of interaction are really necessary for QKD.
Apart from its theoretical importance, there are also practical reasons for addressing this question.
Indeed, two-message protocols are particularly desirable for practical scenarios where parties may drop offline during the protocol execution, may want to send their message at a later point in time, or do not want to keep a state across rounds.
Arguably, this minimal interaction pattern is the property makes traditional cryptographic primitives, such as public-key encryption, so useful.
However, despite this strong motivation and almost fifty years of intense research, the optimal round complexity of QKD is still an open problem.

\ifqcrypt
\subsection*{Our Results}
\else
%-----------------------------------------------------------------------------
\subsection{Our Results}
%-----------------------------------------------------------------------------
\fi
In this work we show that two messages are sufficient to build a QKD protocol with everlasting security.
Specifically, we prove the following statement:
\begin{center}
    If quantum-secure one-way functions exist, then there exists two-message QKD.
\end{center}
Since two messages are clearly necessary for QKD, our protocol achieves the optimal round complexity. Furthermore, only the first message (from Alice to Bob) is quantum, whereas Bob's response is entirely classical.
The protocol satisfies the strong notion of \emph{everlasting security}: As long as the attacker runs in quantum polynomial time during the execution of the protocol, the shared key is hidden in an \emph{information-theoretic} sense.

We view our approach as a big departure from the traditional design of QKD protocols~\cite{BB84,Ekert}, and it is inspired instead by recent works on cryptography with certified deletion~\cite{BK22} and quantum public-key encryption~\cite{QPKE}. Both our protocol and its analysis are entirely elementary and they are simple enough to be fully described in an undergraduate class.
For comparison, it took more than ten years for researchers to establish a formal proof of the first QKD protocol~\cite{BB84}.

Our main technical ingredient is a new framework for building quantum public-key encryption (QPKE), a cryptographic primitive that allows Alice to sample a public key consisting of a quantum state~$\rho$ and a classical string~$\pk$. The scheme is \emph{robust}, in the sense that security is guaranteed to hold even if the distinguisher is given~$\pk$, and it is allowed to tamper arbitrarily with the quantum state~$\rho$. We present two instances of QPKE:
\begin{itemize}
    \item (Everlasting Security) In our first scheme, the message $m$ remains \emph{information-theoretically} hidden, provided that the distinguisher was computationally bounded during the execution of the protocol. This property holds if the distinguisher is given a \emph{single copy} of the public key, which is sufficient to build QKD.
    \item (Computational Security) In our second scheme, the message $m$ is computationally hidden, i.e., we only require security against a computationally bounded distinguisher. While this is a weaker security notion, the advantage of the scheme is that security holds even if the distinguisher is given \emph{arbitrarily many} copies of the public key.
\end{itemize}
The fact that the first construction is only secure in a model where we give access to a single copy of the public key is not a coincidence. It is shown~\cite{QPKE} that given enough (but polynomially-many) copies of the public key, an unbounded adversary can launch a key-recovery attack. This means that there does not exist any QPKE with everlasting security against a distinguisher that sees arbitrarily-many copies of the public key, and it justifies the need for a weaker security notion (computational security).

\paragraph{Everlasting Security.}
We point out a subtle difference between the attacker model that we consider in this work, compared to the standard attacker model for QKD.
The latter, considered for instance in~\cite{shor2000simple,tomamichel2017largely}, models the attacker as a computationally unbounded quantum channel, that is however not allowed to tamper with the information sent over the \emph{classical channels}.
That is, it only assumes the existence of authenticated classical channels, but otherwise does not impose any restriction on the runtime of the distinguisher.
On the other hand, in this work we consider -- in addition to the existence of authenticated classical channels -- an attacker that runs in \emph{quantum polynomial time} during the protocol execution, but it is allowed to be unbounded once the protocol terminates.
That is, we prove \emph{everlasting security} in the sense of~\cite{Everlasting}.

While technically different, we argue that for most practical scenarios the two models are in fact equivalent.
The assumption of an authenticated classical channel is most often justified by having each party sign their own messages with a digital signature, which would require computational assumptions to hold (at the very least) during the execution of the protocol.
In this sense, the mere existence of authenticated classical channels already restricts the attacker to run in quantum polynomial time during the protocol execution (as otherwise it could just break the security of the digital signature).

Finally, it is not hard to show that everlasting security is the \emph{best possible} security notion for QPKE, i.e., there exists a generic attack against any QPKE scheme, if the distinguisher is allowed to run in unbounded time during the execution of the protocol. The attack works even in the presence of authenticated classical channels and succeeds with certainty.
\ifqcrypt
\else
For completeness, we report the proof of this fact in~\cref{appendix}.
\fi

\ifqcrypt
\else
%-----------------------------------------------------------------------------
\subsection{Concurrent Work}
%-----------------------------------------------------------------------------
A concurrent work~\cite{tamper} obtains similar results on robust QPKE, where security holds only against a distinguisher that is allowed to tamper arbitrarily with the quantum portion of the public key. Compared to our work, they only consider the setting of \emph{computational} security, whereas we view the scheme with \emph{everlasting} security as the main contribution of our work, which is the one that enables our two-message QKD protocol. Even focusing on the computational settings, our schemes share many similarities but, interestingly, they are not identical. At a technical level, their approach is based on one-time signatures for Wiesner states, whereas our approach can (in retrospect) be thought of as one-time signing the $\ket{+}$ state. On the other hand,~\cite{tamper} presents a scheme where the public key is a pure state and furthermore their schemes achieve the stronger notion of CCA-security, which we do not consider in this work.

%-----------------------------------------------------------------------------
\subsection{Open Problems}
%-----------------------------------------------------------------------------
Our work leaves open a series of questions, that we hope will inspire further research in this area.
For starters, our protocols are described in the presence of perfect (noiseless) quantum channels.
Any practical protocol would need to withstand the presence of noise.
While theoretically one could simply encode all states using a quantum error correcting code, this may lead to poor concrete efficiency.
We leave open the question of investigating variants our non-interactive QKD protocol that are efficient in the presence of noise.

A compelling aspect of standard QKD protocols such as~\cite{BB84} is that all quantum states consists of tensor products of single qubits, whereas our protocols require coherent superpositions of many-qubit states.
The former property is desirable, since it allows the experimental realization of the protocol on present-day quantum hardware.
We view the problem of constructing a non-interactive QKD protocol in this qubit-by-qubit model as fascinating research direction, and believe that it might require substantially different techniques from the present approach.
\fi

\ifqcrypt
\subsection*{Overview of the Solution}
\else
%-----------------------------------------------------------------------------
\subsection{Overview of the Solution}
%-----------------------------------------------------------------------------
\fi
Our main technical contribution is a new recipe to construct QPKE, whose definition we recall next. The syntax of QPKE consists of three algorithms:
key generation, encryption, and decryption.
The key generation algorithm produces a classical key pair~$(\sk, \pk)$, along with a quantum state~$\rho$.
The pair $(\pk,\rho)$ makes up the public key.
Given this public key, anyone can compute a classical ciphertext~$\ct$ encrypting a given message~$m$, which can only be decrypted by the owner of the secret key.
In terms of security, we require that the message~$m$ should be \emph{information-theoretically} hidden, even if an {efficient} attacker is allowed to tamper arbitrarily with the quantum state associated with the public key.

\paragraph{Constructing QPKE.}
Our construction of QPKE relies on a conceptually different approach than traditional QKD protocols.
To gain some intuition about our framework it is useful to consider a quantum state and measurement
\[
\frac{\ket{0, B_0} + \ket{1, B_1}}{\sqrt{2}} \quad\text{ and }\quad \{\Pi, \mathsf{Id} - \Pi\},
\]
where $\ket{B_0}$ and $\ket{B_1}$ are two \emph{unforgeable basis states} sampled by the key generation algorithm, and~$\{\Pi, \mathsf{Id}-\Pi\}$ is an efficiently-implementable projective measurement that verifies that the second register is in the subspace of all such basis states (without knowing the particular ones used).

It turns out that one can use this to transmit information secretly:
The decrypter keeps the the first qubit and measures it in the Hadamard basis, which has the effect of introducing a random relative phase across the two unforgable basis states (this can already done at key generation time).
The encrypter measures the second register in the Hadamard basis to obtain correlated randomness, that they can use to securely mask a message with a (classical) one-time pad.

To gain an intuition on why this scheme is secure, observe that handing the second register to the adversary is information-theoretically indistinguishable from handing them a random one of the two basis state~$\ket{B_0}$ and~$\ket{B_1}$.
Because those states are assumed to be hard to forge, all the adversary can do is to pass them on (unless they want to make the verification projection fail), and now a Hadamard measurement yields an information-theoretically uniformly random outcome.

Our actual scheme realizes the above abstract template using one-time signatures: The hidden basis are simply valid message-signature pairs, whereas the projection $\Pi$ can be efficiently implemented by running the verification algorithm coherently and measuring the resulting output bit. For more details we refer the reader
\ifqcrypt
to the technical manuscript.
\else
to the technical sections.
\fi

\paragraph{From QPKE to QKD.} Once we built QPKE it appears to be an easy exercise to construct a two-message QKD protocol: Alice can sample a public and send it to Bob, who replies with an encryption of a randomly sampled key $k\in \{0,1\}^\lambda$. However, there is a subtle aspect in the analysis of this protocol, which relates to the fact that secrecy does not necessarily imply security. In fact an attacker may be able to cause Alice and Bob to agree on \emph{different keys} without violating the security of QPKE. We once again rely on one-time signature to (provably) prevent this class of attacks.

\ifqcrypt
\else
%-----------------------------------------------------------------------------
\subsection{Organization of this Paper}
%-----------------------------------------------------------------------------

{In \cref{sec:prelims}, we discuss preliminaries in quantum information and cryptography and we prove some useful technical statements.
In \cref{sec:QPKE}, we define the notion of QPKE, and we present two constructions (with different tradeoffs) whose security can be reduced to the one-wayness of any post-quantum one-way function. In \cref{sec:qkd} we present the formal description and the analysis of our two-message QKD protocol.

%=============================================================================
\section{Preliminaries}\label{sec:prelims}
%=============================================================================
Throughout this work, we denote the security parameter by~$\lambda$.
We denote by $1^\lambda$ the all-ones string of length~$\lambda$.
We say that a function $\negl$ is \emph{negligible} in the security parameter~$\lambda$ if~$\negl(\lambda) = \lambda^{\omega(1)}$.
For a finite set~$S$, we write~$x \gets S$ to denote that $x$ is sampled uniformly at random
from~$S$.
We write $\mathsf{Tr}$ for the trace of a matrix or operator.

%-----------------------------------------------------------------------------
\subsection{Quantum Information}
%-----------------------------------------------------------------------------
In this section, we provide some preliminary background on quantum information.
For a more in-depth introduction, we refer the reader to~\cite{DBLP:books/daglib/0046438}.
A \emph{register}~$\mathsf{x}$ consisting of~$n$ qubits is given by a Hilbert space~$(\mathbb C^2)^{\ot n}$ with name or label~$\mathsf{x}$.
Given two registers~$\mathsf{x}$ and~$\mathsf{y}$, we write~$\mathsf{x} \ot \mathsf{y}$ for the composite register, with Hilbert space the tensor product of the individual registers' Hilbert spaces.
A \emph{pure quantum state} on register~$\mathsf{x}$ is a unit vector~$\ket{\Psi}_\mathsf{x} \in (\mathbb C^2)^{\ot n}$.
A \emph{mixed quantum state} on register~$\mathsf{x}$ is represented by a density operator~$\rho_\mathsf{x}$ on $(\mathbb C^2)^{\ot n}$, which is a positive semi-definite Hermitian matrix with trace~$1$.
Any pure state~$\ket\Psi_{\mathsf x}$ can also be regarded as a mixed state~$\rho_{\mathsf x} = \ket\Psi_{\mathsf x}\bra\Psi_{\mathsf x}$, but there are mixed states that are not pure.
In the above, we use subscripts to denote registers, but we often omit these when clear from context.
We adopt the convention that
\[
\left\{\ket{0}, \ket{1}\right\} \text{ and }\left\{\frac{\ket{0} + \ket{1}}{\sqrt{2}}, \frac{\ket{0} - \ket{1}}{\sqrt{2}}\right\}
\]
denote the \emph{computational} and the \emph{Hadamard basis} states, respectively.

A \emph{quantum channel}~$F$ is a completely-positive trace-preserving (CPTP) map from a register~$\mathsf{x}$ to a register~$\mathsf{y}$.
That is, on input any density matrix $\rho_\mathsf{x}$, the operation $F$ produces $F(\rho_\mathsf{x}) = \tau_\mathsf{y}$, another state on register~$\mathsf{y}$, and the same is true when we apply~$F$ to the $\mathsf x$ register of a quantum state~$\rho_{\mathsf x\mathsf z}$.
For any unitary operator~$U$, meaning $U^\dagger U = UU^\dagger = \mathsf{Id}$, one obtains a quantum channel that maps input states~$\rho$ to output states~$\tau := U \rho U^\dagger$.
The Pauli operators $\mathsf{X},\mathsf{Y},\mathsf{Z}$ are $2\times 2$ matrices that are unitary and Hermitian.
More specifically:
\[\mathsf{X}=\begin{pmatrix} 0 &1 \\  1 &0 \end{pmatrix}, \quad
\mathsf{Y}=\begin{pmatrix} 0 &-i \\  i &0 \end{pmatrix}, \quad
\mathsf{Z}=\begin{pmatrix} 1 &0 \\ 0 &-1 \end{pmatrix}.\]
A \emph{projector}~$\Pi$ is a Hermitian operator such that~$\Pi^2 = \Pi$.
A \emph{projective measurement} is given by a collection of projectors~$\{\Pi_j\}_j$ such that $\sum_j \Pi_j = \mathsf{Id}$.
Given a state~$\rho$, the measurement yields outcome~$j$ with probability~$p_j = \mathsf{Tr}(\Pi_j \rho)$, upon which the state changes to~$\Pi_j \rho \Pi_j / p_j$ (this can be modeled by a quantum channel, but we will not need this).
For any two registers $\mathsf{x}$ and $\mathsf{y}$, the partial trace~$\mathsf{Tr}_\mathsf{y}$ is the unique channel from $\mathsf{x}\otimes\mathsf{y}$ to $\mathsf{x}$ such that~$\mathsf{Tr}_\mathsf{y}(\rho_\mathsf{x} \ot \tau_\mathsf{y}) = \mathsf{Tr}_\mathsf{y}(\tau_\mathsf{y}) \rho_\mathsf{x}$ for all~$\rho_\mathsf{x}$ and~$\tau_\mathsf{y}$.

The \emph{trace distance} between two states $\rho$ and $\tau$, denoted by $\mathsf{Td}(\rho, \tau)$ is defined as
\[
\mathsf{Td}(\rho, \tau) = \frac{1}{2}\|\rho -\tau\|_1 = \frac{1}{2}\mathsf{Tr}\left(\sqrt{(\rho-\tau)^\dagger (\rho-\tau)}\right).
\]
The operational meaning of the trace distance is that $\frac12(1+\mathsf{Td}(\rho,\tau))$ is the maximal probability that two states~$\rho$ and~$\tau$ can be distinguishdd by any (possibly unbounded) quantum channel or algorithm.
{If $\tau=\proj\Phi$ is a pure state, we have the following version of the Fuchs-van de Graaf inequalities:
\begin{align}\label{eq:fuchs}
  1 - \angles{\Phi|\rho|\Phi} \leq \mathsf{Td}(\rho, \tau) \leq \sqrt{1 - \angles{\Phi|\rho|\Phi}}.
\end{align}}

\paragraph{Quantum Algorithms.}
A non-uniform \emph{quantum polynomial-time (QPT) machine} $\{\adv_\lambda\}_{\lambda\in\mathbb{N}}$ is a family of polynomial-size quantum machines $\adv_\lambda$, where each is initialized with a polynomial-size advice state $\ket{\alpha_\lambda}$.
Each $\adv_\lambda$ can be described by a CPTP map.
A quantum interactive machine is simply a sequence of quantum channels, with designated input, output, and work registers.
We say that two probability distributions~$\mathcal{X}$ and $\mathcal{Y}$ are \emph{computationally indistinguishable} if there exists a negligible function $\negl$ such that for all QPT algorithms $\adv_\lambda$ it holds that
\[
\bigl|\Pr\left[1 \gets \adv_\lambda(x) : x \gets \mathcal{X}\right] - \Pr\left[1 \gets \adv_\lambda(y) : y \gets \mathcal{Y}\right]\bigr| = \negl(\lambda).
\]
We say that they are \emph{statistically indistinguishable} if the same holds for all (possibly unbounded) algorithms.

\paragraph{Distinguishing Implies Swapping.}
We recall the formal statement of the equivalence between distinguishing states and swapping on the conjugate basis. This was  proven in~\cite{distswap} and below we show a rephrased version borrowed from~\cite{HMY22}. We actually only state one direction of the implication (the converse is also shown to be true in~\cite{distswap}), since it is the one needed for our purposes.
\begin{theorem}[Distinguishing Implies Swapping~\cite{distswap}]\label{lemma:dis}
    Let $\ket{\Psi}$ and $\ket{\Phi}$ be orthogonal $n$-qubit states, and suppose that a QPT distinguisher $\adv_\lambda$ distinguishes $\ket{\Psi}$ and $\ket{\Phi}$ with advantage $\delta$ without using any ancilla qubits. Then, there exists a polynomial-time computable unitary $U$ over $n$-qubit states such that
    \[
    \frac{\left| \bra{y}U\ket{x} + \bra{x} U \ket{y}\right|}{2} = \delta \text{ where } \ket{x} = \frac{\ket{\Psi} + \ket{\Phi}}{\sqrt{2}} \text{ and } \ket{y} = \frac{\ket{\Psi} - \ket{\Phi}}{\sqrt{2}}.
    \]
    Moreover, if $\adv_\lambda$ does not act on some qubits, then $U$ also does not act on those qubits.
\end{theorem}

\subsection{Information Theory}

Recall the definition of the min-entropy of a random variable $X$ as
\[
H_\infty (X) = -\log\left(\mathsf{max}_x \Pr[X=x]\right).
\]
We recall the definition of average conditional min-entropy in the following.
\begin{definition}[Average Conditional Min-Entropy] Let $X$ be a random-variable supported on a finite set $\mathcal{X}$ and let $Z$ be a (possibly correlated) random variable supported on a finite set $\mathcal{Z}$. The average-conditional min-entropy $\Tilde{H}_\infty(X|Z)$ is defined as
\[
\Tilde{H}_\infty(X|Z) = -\log\left(\mathbb{E}_z\left[\mathsf{max}_{x\in\mathcal{X}} \Pr[X=x|Z=z]\right]\right).
\]    
\end{definition}
It is shown in~\cite{DRS04,DORS08} that the average conditional min-entropy satisfies a \emph{chain rule}, that is 
\begin{equation}\label{eq:chain}
    \tilde{H}_\infty(X | Z) \geq H_\infty(X) - H_0(Z)
\end{equation}
where $H_0(Z)$ denotes the logarithm of the size of the support of $Z$. 
Next, we recall the definition of a seeded randomness extractor.
 \begin{definition}[Extractor]\label{def:extractor}
      A function $\mathsf{Ext} : \{0, 1\}^d \times \mathcal{X} \to \{0, 1\}^\ell$ is called a seeded strong average-case $(k,\varepsilon)$-extractor, if it holds for all random variables $X$ with support $\mathcal{X}$ and $Z$ defined on some finite support that if $\Tilde{H}_\infty(X|Z) \geq k$, then it holds that the statistical distance of the following distributions is a most $\varepsilon$
\[
(\mathsf{seed}, \mathsf{Ext}(\mathsf{seed},X), Z) \approx_\varepsilon (\mathsf{seed}, U, Z) 
\]
where $\mathsf{seed} \gets \{0,1\}^d$ and $U\gets \{0,1\}^\ell$.
 \end{definition}
Recall that a hash function $\mathsf{Hash} : \mathcal{X} \to \mathcal{Y}$ is a universal hash if for all $x \neq x' \in\mathcal{X}$ it holds that
\[
\Pr[\mathsf{Hash}(x) = \mathsf{Hash}(x')] \leq \frac{1}{|\mathcal{Y}|}
\]
where the probability is taken over the choice of the hash function. It is shown~\cite{DRS04,DORS08} that any universal hash function is an average-case randomness extractor.
\begin{lemma}[Leftover Hash Lemma]\label{lemma:lhl}
    Let $X$ be a random-variable supported on a finite set $\mathcal{X}$ and let $Z$ be a (possibly correlated) random variable supported on a finite set $\mathcal{Z}$ such that $\Tilde{H}_\infty(X|Z) \geq k$. Let $\mathsf{Hash}: \mathcal{X} \to \{0,1\}^\ell$, where $\ell \leq k -2 \log\left(\frac{1}{\varepsilon}\right)$, be a family of universal functions. Then $\mathsf{Hash}$ is a seeded strong average-case $(k,\varepsilon)$-extractor.
\end{lemma}

%-----------------------------------------------------------------------------
\subsection{Pseudorandom Functions}\label{subsec:prf}
%-----------------------------------------------------------------------------
We recall the notion of a pseudorandom function~\cite{GGM}. A pseudorandom function (PRF) is a keyed function 
\[
\PRF: \{0,1\}^\lambda \times \{0,1\}^\lambda \to \{0,1\}^\lambda 
\]
that is computationally indistinguishable from a truly random function. More precisely, we require that there exists a negligible function $\negl$ such that for all QPT~$\{\adv_\lambda\}$, for all~$\lambda\in \N$, it holds that the following distributions are computationally indistinguishable
\[
\adv_\lambda(1^\lambda)^{\PRF(k, \cdot)} \approx \adv_\lambda(1^\lambda)^{f(\cdot)}
\]
where $k\gets \{0,1\}^\lambda$ and $f$ is a uniformly-sampled function. It is well-known that quantum-secure PRFs can be built from any one-way function~\cite{GGM}.

%-----------------------------------------------------------------------------
\subsection{One-Time Signatures}\label{subsec:signatures}
%-----------------------------------------------------------------------------
We recall the notion of a one-time signature scheme~\cite{Goldreich}.

\begin{definition}[One-Time Signature]\label{def:sig}
A \emph{one-time signature (OTS)} scheme is defined as a tuple of algorithms~$(\SGen,\Sign,\Ver)$ such that:
\begin{itemize}
\item $(\vk,\sk) \gets \SGen(1^{\lambda})$: A polynomial-time algorithm which, on input the security parameter $1^{\lambda}$, outputs two bit strings~$\vk$ and~$\sk$.
\item $\sigma \gets \Sign(\sk, m)$: A polynomial-time algorithm which, on input the signing key $\sk$ and a message $m$, outputs signature $\sigma$.
\item $\{0,1\} \gets \Ver(\vk, m, \sigma)$: A polynomial-time algorithm which, on input the verification key $\vk$, a message $m$, and a signature $\sigma$, returns a bit denoting accept or reject.
\end{itemize}
\end{definition}

\noindent The OTS scheme is \emph{correct} if for all $\lambda\in \N$ and all messages $m$ it holds that
\begin{align*}
    \Pr\left[ 1 = \Ver(\vk, m, \Sign(\sk, m)) : (\vk,\sk) \gets \SGen(1^{\lambda}) \right] = 1.
\end{align*}

Next we define the notion of strong existential unforgeability for OTS, which states that one should not be able to produce a different signature (even if on the same message) than the one provided by the signing oracle.
It is well-known that strongly unforgeable signatures can be constructed from any one-way function (OWF)~\cite{Goldreich}.
For convenience we define a slightly weaker notion, where the message to be signed is fixed in advance -- this notion is clearly implied by the standard one, where the attacker can query the signing oracle adaptively.

\begin{definition}[Strong Existential Unforgeability]\label{def:unforge}
We say that an OTS scheme $(\SGen,\Sign,\Ver)$ satisfies \emph{{(quantum-secure)} strong existential unforgeability} if there exists a negligible function $\negl$ such that for all QPT~$\{\adv_\lambda\}$, for all~$\lambda\in \N$, and for all messages~$m$, it holds that
\[
    \Pr\left[ 1 = \Ver(\vk, m^*, \sigma^*) \text{ and } (m^*, \sigma^*) \neq (m, \sigma) :
    \begin{array}{l}
         (\vk,\sk) \gets \SGen(1^{\lambda});  \\
         \sigma \gets \Sign(\sk, m);\\
         (m^*, \sigma^*) \gets \adv_\lambda(\vk, m, \sigma)
    \end{array}
       \right] = \negl(\lambda).
\]
\end{definition}

\paragraph{Indistinguishability of Signature States.} We provide a formal statement and a proof of the indistinguishability of our signature states and the corresponding classical mixture.  This proof is inspired by, and closely follows, the work of~\cite{HMY22}.

\begin{lemma}\label{lmm:OTS}
Let $(\SGen,\Sign,\Ver)$ be an OTS scheme that satisfies strong existential unforgeability. Then the following distribution ensambles are computationally indistinguishable
\begin{align*}
\left\{\frac{\ket{0, \sigma_0} + \ket{1, \sigma_1}}{\sqrt{2}} \frac{\bra{0, \sigma_0} + \bra{1, \sigma_1}}{\sqrt{2}}, (\vk_0, \vk_1)\right\} \approx
\left\{\frac{\ketbra{0, \sigma_0}{0, \sigma_0} + \ketbra{1, \sigma_1}{1, \sigma_1}}{{2}}, (\vk_0, \vk_1)\right\}
\end{align*}
where $(\vk_b,\sk_b)\gets \SGen(1^{\lambda})$ and $\sigma_b \gets \Sign(\sk_b, b)$, for $b\in\{0,1\}$.
\end{lemma}
\begin{proof}
By convexity, it suffices to show that no QPT adversary acting on registers $\mathsf{v}$ and $\mathsf{x}$ can distinguish between the states $\ket{\Psi_0}$ and $\ket{\Psi_1}$ with non-negligible probability, where
\[
\ket{\Psi_b} = \sum_{(\vk, \sk)} \sqrt{D(\vk, \sk)} \ket{\vk, \sk}_{\mathsf{s}} \ket{\vk}_\mathsf{v} \otimes \frac{\ket{0,\sigma_0}_\mathsf{x} + (-1)^b \ket{1, \sigma_1}_\mathsf{x} }{\sqrt{2}}
\]
and we adopt the following convention for the notation:
\[
\sk = \{\sk_0,\sk_1\} \text{ ; } \vk = \{\vk_0,\vk_1\} \text{ ; } \sigma_b = \Sign(\sk_b, b) \text{ and  } D(\vk, \sk) = \Pr[(\vk, \sk) = \SGen(1^\lambda)].
\]
Assume towards contradiction that there exists a QPT distinguisher acting on registers $\mathsf{v}$, $\mathsf{x}$, as well as an auxiliary register $\ket{\alpha_\lambda}_\mathsf{a}$ that succeeds with probability $\delta$. Then, by~\Cref{lemma:dis} there exists a polynomial-time computable unitary $U$ such that
\[
\frac{1}{2} \left|
\begin{array}{l}
     \bra{\Psi_1'}_{\mathsf{s}, \mathsf{v}, \mathsf{x}}\bra{\alpha_\lambda}_{\mathsf{a}} (U_{\mathsf{v}, \mathsf{x}, \mathsf{a}} \otimes \mathsf{Id}_{\mathsf{s}})\ket{\Psi_0'}_{\mathsf{s}, \mathsf{v}, \mathsf{x}}  \ket{\alpha_\lambda}_{\mathsf{a}}\\
    +\bra{\Psi_0'}_{\mathsf{s}, \mathsf{v}, \mathsf{x}}\bra{\alpha_\lambda}_{\mathsf{a}} (U_{\mathsf{v}, \mathsf{x}, \mathsf{a}} \otimes \mathsf{Id}_{\mathsf{s}})\ket{\Psi_1'}_{\mathsf{s}, \mathsf{v}, \mathsf{x}}  \ket{\alpha_\lambda}_{\mathsf{a}}
\end{array}
\right| =\delta
\]
where
\[
\ket{\Psi_b'} = \frac{\ket{\Psi_0} + (-1)^b\ket{\Psi_1}}{\sqrt{2}} = \sum_{(\vk, \sk)} \sqrt{D(\vk, \sk)} \ket{\vk, \sk}_{\mathsf{s}} \ket{\vk}_\mathsf{v} \otimes \ket{b, \sigma_b}_\mathsf{x}.
\]
Consequently, it must be the case that either
\begin{itemize}
    \item $\bra{\Psi_1'}_{\mathsf{s}, \mathsf{v}, \mathsf{x}}\bra{\alpha_\lambda}_{\mathsf{a}} (U_{\mathsf{v}, \mathsf{x}, \mathsf{a}} \otimes \mathsf{Id}_{\mathsf{s}})\ket{\Psi_0'}_{\mathsf{s}, \mathsf{v}, \mathsf{x}}  \ket{\alpha_\lambda}_{\mathsf{a}} \geq \delta$, or
    \item $\bra{\Psi_0'}_{\mathsf{s}, \mathsf{v}, \mathsf{x}}\bra{\alpha_\lambda}_{\mathsf{a}} (U_{\mathsf{v}, \mathsf{x}, \mathsf{a}} \otimes \mathsf{Id}_{\mathsf{s}})\ket{\Psi_1'}_{\mathsf{s}, \mathsf{v}, \mathsf{x}}  \ket{\alpha_\lambda}_{\mathsf{a}} \geq \delta$.
\end{itemize}
Without loss of generality we assume that the former holds, but the argument works symmetrically also for the latter case. We will show that this leads to a contradiction with a reduction against the one-time unforgeability of OTS. In fact we will consider an even weaker definition where the adversary receives \emph{no signature}.

On input a verification key $\vk_1$ and an advice $\ket{\alpha_\lambda}_\mathsf{a}$ the reduction samples a uniform $(\vk_0, \sk_0) \gets \SGen(1^\lambda)$, and sets $\vk = \{\vk_0, \vk_1\}$. Then it computes $\sigma_0 \gets \Sign(\sk_0, 0)$ and
\[
U_{\mathsf{v}, \mathsf{x}, \mathsf{a}}\ket{\vk}_{\mathsf{v}} \ket{0,\sigma_0}_\mathsf{x} \ket{\alpha_\lambda}_{\mathsf{a}}
\]
and returns the result of a measurement of the $\mathsf{x}$ register in the computational basis.

We now analyze the success probability of the reduction in producing a valid signature, for a fixed key pair $(\vk_1, \sk_1)$. Let us denote by $\Sigma_1$ the set of all valid signatures on $1$ under $\vk_1$, and by $\sigma_1 = \Sign(\sk_1, 1)$. Then we have that the success probability of the reduction equals
\begin{align*}
    \sum_{\sigma_1' \in\Sigma_1}& |\Sigma_1|\cdot\left\| \bra{1, \sigma_1'}_{\mathsf{x}} U_{\mathsf{v}, \mathsf{x}, \mathsf{a}}\ket{\vk}_{\mathsf{v}} \ket{0,\sigma_0}_\mathsf{x} \ket{\alpha_\lambda}_{\mathsf{a}} \right\|^2 \\&\geq
    \left\| \bra{1, \sigma_1}_{\mathsf{x}} U_{\mathsf{v}, \mathsf{x}, \mathsf{a}}\ket{\vk}_{\mathsf{v}} \ket{0,\sigma_0}_\mathsf{x} \ket{\alpha_\lambda}_{\mathsf{a}} \right\|^2
    \\
&\geq
    \left| \bra{\vk}_{\mathsf{v}}\bra{1, \sigma_1}_{\mathsf{x}} \bra{\alpha_\lambda}_{\mathsf{a}} U_{\mathsf{v}, \mathsf{x}, \mathsf{a}}\ket{\vk}_{\mathsf{v}} \ket{0,\sigma_0}_\mathsf{x} \ket{\alpha_\lambda}_{\mathsf{a}} \right|^2
        \\
&=
    \left| \bra{\vk, \sk}_\mathsf{s}\bra{\vk}_{\mathsf{v}}\bra{1, \sigma_1}_{\mathsf{x}} \bra{\alpha_\lambda}_{\mathsf{a}} (U_{\mathsf{v}, \mathsf{x}, \mathsf{a}}\ket{\vk}_{\mathsf{v}}\otimes \mathsf{Id}_\mathsf{s}) \ket{\vk, \sk}_\mathsf{s}\ket{0,\sigma_0}_\mathsf{x} \ket{\alpha_\lambda}_{\mathsf{a}} \right|^2
\end{align*}
where the second inequality follows from the fact that inserting $\bra{\vk}_{\mathsf{v}}$ and $\bra{\alpha_\lambda}_{\mathsf{a}}$ can only decrease the norm. Now, over the random choice of $(\vk, \sk)$ the success probability of the reduction can be lower bounded by
\begin{align*}
    &\mathbb{E}_{(\vk, \sk)}\left[ \left| \bra{\vk, \sk}_\mathsf{s}\bra{\vk}_{\mathsf{v}}\bra{1, \sigma_1}_{\mathsf{x}} \bra{\alpha_\lambda}_{\mathsf{a}} (U_{\mathsf{v}, \mathsf{x}, \mathsf{a}}\ket{\vk}_{\mathsf{v}}\otimes \mathsf{Id}_\mathsf{s}) \ket{\vk, \sk}_\mathsf{s}\ket{0,\sigma_0}_\mathsf{x} \ket{\alpha_\lambda}_{\mathsf{a}} \right|^2\right] \\
    &\geq \left|\mathbb{E}_{(\vk, \sk)}\left[  \bra{\vk, \sk}_\mathsf{s}\bra{\vk}_{\mathsf{v}}\bra{1, \sigma_1}_{\mathsf{x}} \bra{\alpha_\lambda}_{\mathsf{a}} (U_{\mathsf{v}, \mathsf{x}, \mathsf{a}}\ket{\vk}_{\mathsf{v}}\otimes \mathsf{Id}_\mathsf{s}) \ket{\vk, \sk}_\mathsf{s}\ket{0,\sigma_0}_\mathsf{x} \ket{\alpha_\lambda}_{\mathsf{a}}\right] \right|^2 \\
    &\geq \delta
\end{align*}
where the first inequality follows from Jensen's inequality. This contradicts the unforgeability of OTS and concludes our proof.
\end{proof}

%=============================================================================
\section{Quantum Public-Key Encryption}\label{sec:QPKE}
%=============================================================================
In the following we define and construct the central cryptographic primitive of this work, which we refer to as \emph{quantum-public-key encryption}. 

\subsection{Definitions}

The syntax for this primitive is taken almost in verbatim from~\cite{QPKE}, although in this work we consider a stronger notion of security. For notational convenience, we define the primitive for encrypting one-bit messages, but it is easy to generalize the notion and the corresponding construction to multiple bits, via the standard bit-by-bit encryption. Security of the multi-bit construction follows by a standard hybrid argument.
\begin{definition}[QPKE]\label{def:qpke}
A \emph{quantum-public-key encryption (PKE)} scheme is defined as a tuple of algorithms~$(\SKGen, \PKGen,\Enc,\Dec)$ such that:
\begin{itemize}
\item $\sk \gets \SKGen(1^{\lambda})$: A PPT algorithm which, on input the security parameter $1^{\lambda}$ outputs a secret bit string $\sk$.
\item $({\rho}, \pk) \gets \PKGen(\sk)$: A QPT algorithm which, on input the secret key $\sk$, outputs a (possibly mixed) quantum state $\rho$ and a bit strings $\pk$.
\item $\ct \gets \Enc(\rho, \pk, m)$: A QPT algorithm which, on input the public key $(\rho, \pk)$ and a message~$m\in\{0,1\}$, outputs a ciphertext~$\ct$.
\item $m \gets \Dec(\sk, \ct)$: A QPT algorithm which, on input the secret key $\sk$ and the ciphertext $\ct$, outputs a message $m\in\{0,1\}$.
\end{itemize}
\end{definition}
A QPKE scheme $(\SKGen, \PKGen,\Enc,\Dec)$ satisfies \emph{correctness} if for all $\lambda\in \N$ and all $m\in\{0,1\}$ it holds that:
\begin{align*}
    \Pr\left[ m = \Dec(\sk, \ct) : \sk \gets \SKGen(1^{\lambda}) ; ({\rho}, \pk) \gets \PKGen(\sk); \ct \gets \Enc(\rho, \pk, m) \right] = 1.
\end{align*}

\paragraph{Everlasting Security.} Next, we define the security notion of \emph{everlasting security} for QPKE. Informally, we require that the message is unconditionally hidden from the eyes of an attacker, even if a QPT attacker is allowed to tamper with the public key arbitrarily. However, the attacker is supplied a \emph{single copy} of the public key.
\begin{definition}[Everlasting Security]\label{def:strong_ever}
For a family of QPT algorithms $\{\adv_\lambda \}_{\lambda\in\N}$, we define the experiment $\mathsf{Exp}^{\adv_\lambda}(1^\lambda, m)$ as follows:
\begin{enumerate}
    \item Sample $\sk \gets \SKGen(1^\lambda)$ and $({\rho}, \pk) \gets \PKGen(\sk)$ and send the corresponding public key $(\rho,\pk)$ to $\adv_\lambda$.
    \item $\adv_\lambda$ returns two quantum registers. The first register is parsed as the modified public-key register, whereas the second register is arbitrary and will be referred to as the adversary's~internal~register.
    \item Compute $\ct$ by applying the map defined by $\Enc(\cdot, \pk, m)$ to the public-key register returned by the adversary in the previous round.
    \item The output of the experiment is defined to be the joint state of $\ct$ and the internal register of the adversary.
\end{enumerate}
Then we say that a QPKE scheme $(\SKGen, \PKGen,\Enc,\Dec)$ satisfies \emph{everlasting security} if there exists a negligible function $\negl$ such that for all $\lambda\in \N$ and all QPT $\adv_\lambda$ it holds that
\[
\mathsf{Td}\left(\mathsf{Exp}^{\adv_\lambda}(1^\lambda, 0), \mathsf{Exp}^{\adv_\lambda}(1^\lambda, 1)\right) = \negl(\lambda).
\]
\end{definition}
Let us comment on the definition as stated above.
First, we remark that the definition can be easily extended to the case of multi-bit messages, provided that the syntax of the encryption scheme is extended accordingly. We also mention that an alternative definition might also allow the adversary to do some arbitrary post-processing on the output of the experiment. However, our definition is equivalent (and arguably simpler) by the monotonicity of the trace distance.

An important point of our definition (which distinguishes it from prior works) is that the attacker is allowed to modify the quantum states arbitrarily, although it cannot tamper with the classical information (such as~$\pk$ or~$\ct$).
This models the presence of \emph{authenticated classical channels}, which are assumed to deliver the classical information faithfully.
Note that the same assumption is also present (although somewhat more implicitly) for the standard notion of \emph{classical} PKE, where the encryption algorithm in the CPA/CCA-security experiment is always provided as input the correct public key sampled by the challenger.
This restriction is of course necessary, since if the attacker is allowed to choose the public key arbitrarily, then the definition would be impossible to achieve.

Finally, we mention that a stronger definition would allow the adversary to see a polynomial number of copies of the public key, instead of a single one. Unfortunately the work of~\cite{QPKE} shows a key recovery attack against any QPKE, if the attacker is given sufficiently many copies of the public key and it is allowed to run in unbounded time. This immediately rules out any QPKE with everlasting security in the presence of polynomial copies of the public key, since an unbounded distinguisher can simply run such key recovery algorithm, and decrypt the challenge ciphertext using the honest decryption algorithm. To overcome this limitation, we define the notion of computational security.

\paragraph{Computational Security.} We define the weaker notion of computational security for QPKE, where the message is only required to be kept hidden against computationally bounded adversary. The upshot is that this can hold even if the adversary is given access to multiple copies of the public key. We present a formal definition below.
\begin{definition}[Computational Security]\label{def:strong_comp}
For a family of QPT algorithms $\{\adv_\lambda \}_{\lambda\in\N}$, we define the experiment $\mathsf{Exp}^{\adv_\lambda}(1^\lambda, m, n)$ as follows:
\begin{enumerate}
    \item Sample $\sk \gets \SKGen(1^\lambda)$ and 
    $\left\{({\rho}_i, \pk_i) \gets \PKGen(\sk)\right\}_{i = 1, \dots, n}$ and send the corresponding public keys $(\rho_i,\pk_i)$ to $\adv_\lambda$.
    \item $\adv_\lambda$ returns two quantum registers. The first register is parsed as the modified public-key register, whereas the second register is arbitrary and will be referred to as the adversary's~internal~register.
    \item Compute $\ct$ by applying the map defined by $\Enc(\cdot, \pk_1, m)$ to the public-key register returned by the adversary in the previous round.
    \item The output of the experiment is defined to be the joint state of $\ct$ and the internal register of the adversary.
\end{enumerate}
Then we say that a QPKE scheme $(\SKGen, \PKGen,\Enc,\Dec)$ satisfies \emph{computational security} if there exists a negligible function $\negl$ such that for all $\lambda\in \N$, all polynomials $n = n(\lambda)$, and all QPT $\adv_\lambda$ it holds that the distributions
\[
\mathsf{Exp}^{\adv_\lambda}(1^\lambda, 0, n) \approx \mathsf{Exp}^{\adv_\lambda}(1^\lambda, 1, n) 
\]
are computationally indistinguishable.
\end{definition}

\subsection{Everlasting Secure QPKE}

We describe our scheme below.
As the only computational ingredient, we assume the existence of a {quantum-secure} strongly existentially unforgeable one-time signature scheme $(\SGen, \Sign, \Ver)$, see \cref{subsec:signatures}.
{As discussed, this can be constructed from any quantum-secure one-way function.}

\begin{itemize}
\item $\SKGen(1^{\lambda})$:
\begin{itemize}
\item Sample two key pairs  $(\sk_{0}, \vk_{0}) \gets \SGen(1^\lambda)$ and $(\sk_{1}, \vk_{1}) \gets \SGen(1^\lambda)$.
\item Compute $\sigma_{0} \gets \Sign(\sk_{0}, 0)$ and $\sigma_{1} \gets \Sign(\sk_{1}, 1)$.
\item Sample a bit $d_0 \gets \{0,1\}$.
\item Return $\sk = (\vk_0, \vk_1, \sigma_{0}, \sigma_{1}, d_0)$.
\end{itemize}

\item $\PKGen(\sk)$:
\begin{itemize}
\item Define the state \[\ket{\Psi} = \frac{\ket{0,\sigma_{0}} + (-1)^{d_0}\ket{1,\sigma_{1}}}{\sqrt{2}}.\]
This state is efficiently computable by preparing an EPR pair and CNOT-ing the bits of the signatures into an auxiliary register, controlled on the value of the first qubit. The relative phase can be then added by a controlling the application of $\mathsf{Z}$ with $d_0$.
\item Set the quantum part of the public key~$\rho$ to be the state $\ket{\Psi}$ and set the classical part of the public key and the classical secret key to  $\pk = (\vk_{0}, \vk_{1})$.
\end{itemize}

    \item $\Enc(\rho, \pk, m)$:
    \begin{itemize}
        \item Project $\rho$ onto the subspace of valid signatures of $0$ and $1$, under $\vk_{0}$ and $\vk_{1}$, respectively.
        More precisely, denote by~$\Sigma_{0}$ and~$\Sigma_{1}$ the set of accepting signatures on~$0$ and~$1$, under~$\vk_{0}$ and~$\vk_{1}$, respectively, and consider the projector
        \[
        \Pi = \sum_{\sigma \in \Sigma_{0}} \proj{0,\sigma} + \sum_{\sigma \in \Sigma_{1}} \proj{1,\sigma}.
        \]
        Apply the projective measurement~$\{\Pi, \mathsf{Id} - \Pi\}$, and abort the execution {(return~$\perp$)} if the measurement returns the second outcome.
        Note that this measurement can be implemented efficiently by running the verification algorithm coherently, CNOT-ing the output bit on a separate register and measuring it.

        \item Measure the residual state in the Hadamard basis, to obtain a bit string $(d_1, d_2)$, where we denote by~$d_1\in\{0,1\}$ the first bit of the measurement outcome and by $d_2$ the rest.
        \item Return the following as the classical ciphertext:
        \begin{equation}\label{eq:ct}
        \ct = (m\oplus d_1, d_2).
        \end{equation}
    \end{itemize}
    \item $\Dec(\sk, \ct)$:
    \begin{itemize}
        \item Parse $\ct = (\ct_1, \ct_2)$, where $\ct_1 \in \{0,1\}$ is one bit, and return
        \begin{equation}\label{eq:m}
          m = d_0 \oplus \ct_1 \oplus \ct_2 \cdot (\sigma_0 {\oplus} \sigma_1).
        \end{equation}
    \end{itemize}
\end{itemize}

\paragraph{Analysis.} We claim that the scheme satisfies correctness. First, observe that the state $\rho$ taken as input by the encryption algorithm in the image of the projector~$\Pi$ as defined above. Consequently, applying the projective measurement $\{\Pi, \mathsf{Id} - \Pi\}$ returns the outcome associated with~$\Pi$ with certainty and does not change the state. Applying the Hadamard transformation to the state $\ket{\Psi}$ gives
\[
\mathsf{H} \ket{\Psi} \propto \sum_{d_1, d_2} (-1)^{(d_1,d_2) \cdot (0, \sigma_{0})} \ket{d_1, d_2} + (-1)^{d_0 \oplus (d_1, d_2) \cdot (1, \sigma_{1})} \ket{d_1, d_2}  = \sum_{d_1, d_2 \;:\; (d_1, d_2) \cdot (1, \sigma_{0} \oplus \sigma_{1}) = d_0} \ket{d_1, d_2},
\]
omitting overall normalization factors.
Therefore, a measurement returns a uniformly random bit string~$(d_1,d_2)$ satisfying 
\begin{align*}
  d_1 \oplus d_2 \cdot (\sigma_{0} \oplus \sigma_{1})&= d_0.
\end{align*}
Substituting \cref{eq:ct} in \cref{eq:m} and using this relation, we obtain
\begin{align*}
  d_0 \oplus \ct_1 \oplus \ct_2 \cdot (\sigma_0 \oplus \sigma_1)
= d_0 \oplus (m\oplus d_1) \oplus d_2 \cdot (\sigma_0 \oplus \sigma_1)
= m,
\end{align*}
as desired.
Next, we show that the scheme satisfies everlasting security.
\begin{theorem}[Everlasting security]\label{thm:main}
If quantum-secure one-way functions exist, then the QPKE $(\SKGen, \PKGen,\Enc,\Dec)$ satisfies everlasting security.
\end{theorem}
\begin{proof}
We proceed by defining a series of hybrid experiments that we show to be indistinguishable from the eyes of any (possibly unbounded) algorithm.
For convenience, we define \[\mathsf{Adv}(i) = \mathsf{Td}\left(\mathsf{Hyb}_i^{\adv_\lambda}(1^\lambda, 0), \mathsf{Hyb}_i^{\adv_\lambda}(1^\lambda, 1)\right).\]
    \begin{itemize}
        \item $\mathsf{Hyb}_0^{\adv_\lambda}(1^\lambda, b)$: This is the original experiment $\mathsf{Exp}^{\adv_\lambda}(1^\lambda, b)$, as defined in~\Cref{def:strong_ever}.
        \item $\mathsf{Hyb}_1^{\adv_\lambda}(1^\lambda, b)$: In this experiment, we modify the $\PKGen$ algorithm to measure the state $\ket{\Psi}$ in the computational basis, before outputting $\rho$.
    \end{itemize}
Since the result $\sk$ of the $\SKGen$ algorithm is not used in the experiment, we only need to argue that the reduced states of $(\pk,\rho)$ are unchanged by this modification.
This is indeed the case, since adding a random phase is equivalent to measuring in the computational basis by a standard Pauli $\mathsf{Z}$-twirl argument. Thus the two experiments are identical from the perspective of the adversary and therefore $\mathsf{Adv}(0) = \mathsf{Adv}(1)$.
    \begin{itemize}
        \item $\mathsf{Hyb}_2^{\adv_\lambda}(1^\lambda, b)$: In this experiment, we further modify the $\PKGen$ algorithm to sample the state~$\rho$ as follows.
        Flip a random coin $c\gets \{0,1\}$.
        If $c=0$ then return~$\ketbra{0,\sigma_0}{0,\sigma_0}$, and otherwise return~$\ketbra{1,\sigma_1}{1,\sigma_1}$.
    \end{itemize}
    Observe that the state $\rho$ returned by the modified $\PKGen$ algorithm is the classical mixture
    \[
\rho = \frac{\ketbra{0,\sigma_0}{0,\sigma_0} + \ketbra{1,\sigma_1}{1,\sigma_1}}{2}
    \]
    which is identical to the state returned in the previous hybrid.
    Therefore $\mathsf{Adv}(1) = \mathsf{Adv}(2)$.
    Next, let us denote by~$\rho^*$ the reduced density matrix of the modified public-key register returned by the adversary in step~2 of the experiment.
    We will assume that the state~$\rho^*$ returned by the adversary is such that the encryption algorithm accepts (i.e., does not abort) with non-negligible probability (for otherwise $\mathsf{Adv}(2)=\negl(\lambda)$ and we are done).
    We can then establish that the state~$\rho^*$, once projected onto the image of $\Pi$, must be negligibly close in trace distance from the state produced by the $\PKGen$ algorithm.
    \begin{claim}\label{claim:0}
        There exists a negligible function $\negl$ such that{, for $c\in\{0,1\}$ the result of the coin toss in the $\Gen$ algorithm,}
    \[
    \mathsf{Td} \left(\frac{\Pi\rho^*\Pi}{\mathsf{Tr}\left(\Pi\rho^*\right)} , \ketbra{c, \sigma_c}{c, \sigma_c}\right) = \negl(\lambda).
    \]
    \end{claim}
    \begin{proof}[Proof of \cref{claim:0}]
The proof follows by a reduction to the unforgeability of the OTS.
{Indeed, assume for sake of contradiction that the post-measurement state is non-negligibly far from $\ket{c,\sigma_c}$ in trace distance.
By \cref{eq:fuchs}, the latter is equivalent to saying that if we measure in the computational basis then the probability of obtaining outcome $(c,\sigma_c)$ is non-negligibly smaller than~1.
Since the post-measurement state is supported on range of~$\Pi$, it follows that if we measure in the computational basis then we must with non-negligible probability obtain an outcome~$(z,\sigma)$ such that
\begin{align*}
  (z,\sigma) \neq (c,\sigma_{c}) \quad\text{and}\quad \Ver(\vk_z     , z, \sigma) = 1,
\end{align*}
that is, $(z,\sigma) \neq (c, \sigma_c)$ is a valid message-signature pair.
As the adversary along with the projective measurement~$\{\Pi,I-\Pi\}$ and the standard basis measurement run in quantum polynomial time, and the projective measurement returns~$\Pi$ with non-negligible probability, this contradicts the strong existential unforgeability of the OTS scheme.}
\end{proof}
\noindent
We conclude by establishing that the message $m$ is statistically hidden in the last experiment.
\begin{claim}\label{claim:00}
There exists a negligible function $\negl$ such that $\mathsf{Adv}(2) = \negl(\lambda)$.
\end{claim}
\begin{proof}[Proof of \cref{claim:00}]
{By~\Cref{claim:0}, the post-measurement state is negligibly close to the state $\ket{c, \sigma_c}$.
As the latter is a pure state and extensions of pure states are always in tensor product, by Uhlmann's theorem it follows that the post-measurement register is negligibly close to being in tensor product with the internal register of the adversary.
Thus it suffices to show that the output distribution of the~$\Enc$ algorithm does not depend on the message~$m$ when given as input~$\ket{c, \sigma_c}$.}
To see this, observe that the rotated state in the Hadamard basis is up to overall normalization given by
\[
\mathsf{H} \ket{c, \sigma_c} \propto \sum_{d'}(-1)^{d'\cdot(c, \sigma_c)} \ket{d'},
\]
and therefore a measurement returns a uniformly random bit string~$d' = (d_1, d_2)$.
{In particular, $\ct = (m \oplus d_1, d_2)$ has the same distribution for~$m\in\{0,1\}$.
As explained above, it is also negligibly close to being independent from the internal register of the adversary, and thus the claim follows.}
\end{proof}
\noindent By~\Cref{claim:00} we have that
\[
\mathsf{Adv}(0) = \mathsf{Adv}(1) = \mathsf{Adv}(2) = \negl(\lambda),
\]
and this concludes the proof of \cref{thm:main}.
\end{proof}

\subsection{Computational Secure QPKE}

Our scheme assumes the existence of a {quantum-secure} strongly existentially unforgeable one-time signature scheme $(\SGen, \Sign, \Ver)$ and a quantum-secure pseudorandom function $\PRF$. Both such building blocks can be constructed assuming any quantum-secure one-way function.

\begin{itemize}
\item $\SKGen(1^{\lambda})$:
\begin{itemize}
\item Sample a key $k \gets \{0,1\}^\lambda$ and set $\sk = k$.
\end{itemize}

\item $\PKGen(\sk)$:
\begin{itemize}
\item Sample two uniform $(r_0, r_1) \gets \{0,1\}^\lambda$ and compute 
\[
(\sk_{0}, \vk_{0}) \gets \SGen(1^\lambda; \PRF(k,r_0)) \quad \text{ and } \quad (\sk_{1}, \vk_{1}) \gets \SGen(1^\lambda; \PRF(k,r_1)).
\]
\item Compute $\sigma_{0} \gets \Sign(\sk_{0}, 0)$ and $\sigma_{1} \gets \Sign(\sk_{1}, 1)$.
\item Define the state \[\ket{\Psi} = \frac{\ket{0,\sigma_{0}} + \ket{1,\sigma_{1}}}{\sqrt{2}}.\]
This state is efficiently computable by preparing an EPR pair and CNOT-ing the bits of the signatures into an auxiliary register, controlled on the value of the first qubit.
\item Set the quantum part of the public key~$\rho$ to be the state $\ket{\Psi}$ and set the classical part of the public key and the classical secret key to  $\pk = (\vk_{0}, \vk_{1}, r_0, r_1)$.
\end{itemize}

\item $\Enc(\rho, \pk, m)$:
\begin{itemize}
\item Project $\rho$ onto the subspace of valid signatures of $0$ and $1$, under $\vk_{0}$ and $\vk_{1}$, respectively.
More precisely, denote by~$\Sigma_{0}$ and~$\Sigma_{1}$ the set of accepting signatures on~$0$ and~$1$, under~$\vk_{0}$ and~$\vk_{1}$, respectively, and consider the projector
        \[
        \Pi = \sum_{\sigma \in \Sigma_{0}} \proj{0,\sigma} + \sum_{\sigma \in \Sigma_{1}} \proj{1,\sigma}.
        \]
        Apply the projective measurement~$\{\Pi, \mathsf{Id} - \Pi\}$, and abort the execution {(return~$\perp$)} if the measurement returns the second outcome.
        Note that this measurement can be implemented efficiently by running the verification algorithm coherently, CNOT-ing the output bit on a separate register and measuring it.
        \item Apply the $\mathsf{Z}^m$ operator to the first qubit of $\rho$, classically controlled on the message $m$.
        \item Set $\ct$ to be the residual state, along with $(r_0, r_1)$.
\end{itemize}

\item $\Dec(\sk, \ct)$:
\begin{itemize}
    \item Use the secret key $k$ to recompute 
\[
(\sk_{0}, \vk_{0}) \gets \SGen(1^\lambda; \PRF(k,r_0)) \quad \text{ and } \quad (\sk_{1}, \vk_{1}) \gets \SGen(1^\lambda; \PRF(k,r_1)).
\]
along with $\sigma_{0} \gets \Sign(\sk_{0}, 0)$ and $\sigma_{1} \gets \Sign(\sk_{1}, 1)$.
\item Measure the quantum state of $\ct$ in the 
\[
\left\{\frac{\ket{0,  \sigma_0} + \ket{1,  \sigma_1}}{\sqrt{2}}, \frac{\ket{0,  \sigma_0} - \ket{1,  \sigma_1}}{\sqrt{2}}\right\}
\]
basis. And return the corresponding outcome.
\end{itemize}
\end{itemize}

\paragraph{Analysis.} To see why the scheme satisfies correctness, first observe that the state $\ket{\Psi}$ as defined in the $\PKGen$ algorithm lies in the image of the projector $\Pi$. Therefore the projective measurement $\{\Pi, \mathsf{Id} - \Pi\}$ acts as the identity on $\ket{\Psi}$. Then, the state output by the encryption algorithm corresponds to
\[
(\mathsf{Z}^{m} \otimes \mathsf{Id}) \frac{\ket{0,\sigma_{0}} + \ket{1,\sigma_{1}}}{\sqrt{2}}=\frac{\ket{0,\sigma_{0}} + (-1)^m\ket{1,\sigma_{1}}}{\sqrt{2}}.
\]
Therefore, the output of the decryption algorithm equals $m$ with certainty. Next, we show that the scheme is computationally secure.

\begin{theorem}[Computational security]\label{thm:comp}
If quantum-secure one-way functions exist, then the QPKE $(\SKGen, \PKGen, \Enc,\Dec)$ satisfies computational security.
\end{theorem}
\begin{proof}
We proceed by defining a series of hybrid experiments that we show to be indistinguishable from the eyes of any QPT algorithm.
\begin{itemize}
\item $\mathsf{Hyb}_0^{\adv_\lambda}(1^\lambda, b)$: This is the original experiment $\mathsf{Exp}^{\adv_\lambda}(1^\lambda, b, n)$, as defined in~\Cref{def:strong_comp}.
\item $\mathsf{Hyb}_1^{\adv_\lambda}(1^\lambda, b)$: In this experiment, we simulate the output of the PRF by lazy sampling, i.e., every time that the $\PKGen$ algorithm is invoked, the experiment sample a uniform tuple $(r_0, r_1, \tilde{r}_0, \tilde{r}_1)$ and uses the latter pair as the randomness for the OTS scheme. To keep things consistent, the experiment maintains a list of all such tuples.
\end{itemize}
Note that the only difference between these two hybrids is in the way the random coins of the OTS are sampled. Therefore, by the pseudorandomness of $\PRF$, the two hybrids are computationally indistinguishable. Note that, in the second hybrid, different copies of the public key are now independent from each other.
\begin{itemize}
\item $\mathsf{Hyb}_2^{\adv_\lambda}(1^\lambda, b)$: In this experiment, we further modify the \emph{first invocation} of the $\PKGen$ algorithm to sample the state~$\rho$ as follows. Flip a random coin $c\gets \{0,1\}$. If $c=0$ then return~$\ketbra{0,\sigma_0}{0,\sigma_0}$, and otherwise return~$\ketbra{1,\sigma_1}{1,\sigma_1}$.
\end{itemize}
    Observe that the state $\rho$ returned by the modified $\PKGen$ algorithm is the classical mixture
    \[
\rho = \frac{\ketbra{0,\sigma_0}{0,\sigma_0} + \ketbra{1,\sigma_1}{1,\sigma_1}}{2}.
    \]
By a direct application of~\Cref{lmm:OTS}, we can conclude that the two hybrids are computationally indistinguishable. At this point, we can appeal to~\Cref{claim:0} (in the proof of~\Cref{thm:main}) to establish that the state $\rho^*$, as returned by the adversary in the security experiment, must be within negligible trace distance from a basis state. The proof is concluded by noting that the distributions of an encryption of $0$ and an encryption of $1$ are identical, up to a global phase, if the algorithm is called on input any basis state.
\end{proof}

\input{qkd}
\section*{Acknowledgments}
%=============================================================================
The authors would like to thank Khashayar Barooti for many discussion on quantum public key encryption and Takashi Yamakawa for suggesting a proof of~\Cref{thm:unbounded}.
G.M.~was partially funded by the German Federal Ministry of Education and Research (BMBF) in the course of the 6GEM research hub under grant number 16KISK038 and by the Deutsche Forschungsgemeinschaft (DFG, German Research Foundation) under Germany's Excellence Strategy - EXC 2092 CASA – 390781972.
M.W.~acknowledges support by the the European Union (ERC, SYMOPTIC, 101040907), by the Deutsche Forschungsgemeinschaft (DFG, German Research Foundation) under Germany's Excellence Strategy - EXC\ 2092\ CASA - 390781972, by the BMBF through project QuBRA, and by the Dutch Research Council (NWO grant OCENW.KLEIN.267).
Views and opinions expressed are those of the author(s) only and do not necessarily reflect those of the European Union or the European Research Council Executive Agency.
Neither the European Union nor the granting authority can be held responsible for them.
\fi

\bibliographystyle{splncs04}
\bibliography{ref}

% \ifqcrypt
% \else
\appendix

\section{Impossibility of Unconditionally Secure QPKE}\label{appendix}

In the following we describe a simple argument that rules out the existence of \emph{unconditionally} secure QPKE, even if the adversary is given access to a single copy of the public key. In more details, we show an adversary that can break the security of any QPKE if it is allowed to be unbounded also \emph{during} the protocol execution. The following proof was suggested by Takashi Yamakawa, who should be credited for the argument.
\begin{theorem}[Unconditional Security]\label{thm:unbounded}
    There does not exists an unconditionally secure QPKE.
\end{theorem}
\begin{proof}
    We provide a description of our generic attacker that, running in unbounded time, wins the experiment defined in~\Cref{def:strong_ever} with certainty. The attacker proceeds as follows.
    \begin{itemize}
        \item On input a state $\rho$ and a bitstring $\pk$, enter the following loop:
        \begin{itemize}
            \item Sample a secret key $\sk^* \gets \SKGen(1^\lambda)$ uniformly.
            \item Run $(\rho^*, \pk^*) \gets \PKGen(\sk)$.
            \item If $\pk^* = \pk$ exit the loop and return $(\rho^*, \pk^*, \sk^*)$.
            \item Else, start over.
        \end{itemize}
        \item Let $(\rho^*, \pk^*, \sk^*)$ be the tuple output by the above loop. Return $\rho^*$ to the challenger.
        \item Upon receiving $\ct$, use $\sk^*$ to decrypt the message.
    \end{itemize}
    To show that the attack always succeeds, it suffices to observe that the internal loop eventually returns a tuple $(\rho^*, \pk^*, \sk^*)$ such that
    \[
    \pk^* = \pk \quad \text{ and } \quad (\rho^*, \underbrace{\pk^*}_{=\pk}) = \PKGen(\sk^*)
    \]
    In particular, this means that the algorithm $\Enc$ run by the challenger is run on a valid pair $(\rho^*, \pk)$. Therefore, by the correctness of QPKE, the secret key $\sk^*$ recovers the correct message with certainty.
\end{proof}

\end{document}

%% file: macros.tex
\newcommand{\ct}{\textsf{ct}}

\newcommand{\pk}{\textsf{pk}}
\newcommand{\sk}{\textsf{sk}}

 \newcommand{\Gen}{\textsf{Gen}}
 \newcommand{\SKGen}{\textsf{SKGen}}
 \newcommand{\PKGen}{\textsf{PKGen}}
\newcommand{\Enc}{\textsf{Enc}}
\newcommand{\Dec}{\textsf{Dec}}
\newcommand{\PRF}{\textsf{PRF}}

\newcommand{\msg}{\textsf{msg}}
\newcommand{\resp}{\textsf{resp}}
\newcommand{\st}{\textsf{st}}

 \newcommand{\QKDFirst}{\textsf{QKDFirst}}
  \newcommand{\QKDSecond}{\textsf{QKDSecond}}
   \newcommand{\QKDDecode}{\textsf{QKDDecode}}

%% file: qkd.tex
%-----------------------------------------------------------------------------
\section{Two-Message Quantum Key Distribution}\label{sec:qkd}
%-----------------------------------------------------------------------------

In the following we outline how to use a QPKE scheme $(\SKGen, \PKGen, \Enc, \Dec)$ to construct a QKD protocol with a minimal number of two rounds of interaction, as announced in the introduction.

\subsection{Definitions}
We give a formal definition of quantum key distribution in the everlasting settings, i.e., where an attacker is required to be computationally bounded only during the execution of the protocol. For convenience, we adopt a syntax specific for two-message protocols, but the definitions can be extended to the more general interactive settings canonically.

\begin{definition}[Two-Message QKD]\label{def:2mqkd}
 A \emph{quantum key distribution (QKD)} scheme is defined as a tuple of algorithms~$(\QKDFirst, \QKDSecond, \QKDDecode)$ such that:
\begin{itemize}
\item $(\msg, \mu, \st) \gets \QKDFirst(1^{\lambda})$: A QPT algorithm which, on input the security parameter $1^{\lambda}$ outputs a message, consisting of a classical component $\msg$ and a (possibly mixed) quantum state $\mu$, and an internal state $\st$.
\item $\{(\resp, \eta, k), \bot\} \gets \QKDSecond(\msg, \mu)$: A QPT algorithm which, on input the first message $(\msg, \mu)$, outputs a response, consisting of a classical component $\resp$ and a (possibly mixed) quantum state $\eta$, along with a key $k \in\{0,1\}^\lambda$, or a distinguished symbol $\bot$, denoting rejection.
\item $\{k, \bot\} \gets \QKDDecode(\st, \resp, \eta)$: A QPT algorithm which, on input the internal state $\st$, and the response $(\resp, \eta)$, returns a key $k \in\{0,1\}^\lambda$ or a distinguished symbol $\bot$, denoting rejection.
\end{itemize}
\end{definition}
We say that a QKD scheme $(\QKDFirst, \QKDSecond, \QKDDecode)$ satisfies \emph{correctness} if for all $\lambda\in \N$ it holds that:
\begin{align*}
    \Pr\left[ \bot = \QKDSecond(\msg, \mu) : 
         (\msg, \mu, \st) \gets \QKDFirst(1^{\lambda})
     \right] = 0 
\end{align*}
and
\begin{align*}
    \Pr\left[ k = \QKDDecode(\st, \resp, \eta) : 
    \begin{array}{l}
         (\msg, \mu, \st) \gets \QKDFirst(1^{\lambda}) ;  \\
         (\resp, \eta, k) \gets \QKDSecond(\msg, \mu) 
    \end{array}
     \right] = 1.
\end{align*}

\paragraph{Everlasting Security.} Next, we define the security notion of \emph{everlasting security} for QKD, which consists of two properties. Privacy requires that the key $k$ should be hidden unconditionally in the presence of an adversary that is computationally bounded during the execution of the protocol. In addition, as standard for QKD, we assume the existence of an authenticated classical channel, which is modeled by not allowing the adversary to tamper with the classical messages. On the other hand, verifiability requires that no computationally bounded adversary should be able to cause Alice and Bob to disagree on the key, without any of them noticing.

\begin{definition}[Everlasting Security]\label{def:qkd_ever}
For a family of QPT algorithms $\{\adv_\lambda \}_{\lambda\in\N}$, we define the experiment $\mathsf{QKDSec}^{\adv_\lambda}(1^\lambda)$ as follows:
\begin{enumerate}
\item Sample $(\msg, \mu, \st) \gets \QKDFirst(1^{\lambda})$ and send the corresponding first message $(\msg, \mu)$ to $\adv_\lambda$.
\item $\adv_\lambda$ returns two quantum registers. The first register is parsed as the modified first message, whereas the second register is arbitrary and will be referred to as the adversary's internal~register.
\item Compute $\{(\resp, \eta, k_0), \bot\}$ by applying the map defined by $\QKDSecond(\msg, \cdot)$ to the modified first message register returned by the adversary in the previous round. 
\begin{enumerate}
\item If the above message is $\bot$, set $(k_0, k_1) = (\bot, \bot)$ and conclude the experiment.
\item Otherwise, return $(\resp, \eta)$ to the adversary, along with its internal state.
\end{enumerate}
\item $\adv_\lambda$ returns once again two quantum registers. The first register is parsed as the modified response, whereas the second register is the adversary's internal~register.
\item Compute $\{k_1, \bot\}$ by applying the map defined by $\QKDDecode(\st, \resp, \cdot)$  to the modified response register returned by the adversary in the previous round. If the result is $\bot$, then set $k_1 = \bot$.
\item The output of the experiment is defined to be the internal register of the adversary.
\end{enumerate}
Then we say that a QKD scheme $(\QKDFirst, \QKDSecond, \QKDDecode)$ satisfies \emph{everlasting security} if the following properties hold.
\begin{itemize}
    \item (Privacy) There exists a negligible function $\negl$ such that for all $\lambda\in \N$ and all QPT $\adv_\lambda$ it holds that
\[
\mathsf{Td}\left(\left\{\mathsf{QKDSec}^{\adv_\lambda}(1^\lambda), k_0, k_1\right\}, \left\{\mathsf{QKDSec}^{\adv_\lambda}(1^\lambda), \tilde k_0, \tilde k_1\right\}\right) = \negl(\lambda)
\]
where the variables $k_0$ and $k_1$ are defined in the experiment, whereas $\tilde k_b$, for $b\in\{0,1\}$, is defined as
\[
\begin{cases}
			\tilde k_b = \bot & \text{if $k_b=\bot$}\\
            \tilde k_b \gets \{0,1\}^\lambda & \text{otherwise.}
		 \end{cases}
\]
    \item (Verifiability) There exists a negligible function $\negl$ such that for all $\lambda\in \N$ and all QPT $\adv_\lambda$ it holds that
    \[
    \Pr[k_0 \neq k_1 \text{ and } k_1 \neq \bot] = \negl(\lambda)
    \]
    where $k_0$ and $k_1$ are defined in the experiment.
\end{itemize}
\end{definition}

% \begin{definition}[Verifiability]
%     Let $\mathsf{QKDSec}^{\adv_\lambda}(1^\lambda)$ be the experiment as defined above. Then we say that a QKD scheme $(\QKDFirst, \QKDSecond, \QKDDecode)$ satisfies \emph{verifiability} if there exists a negligible function $\negl$ such that for all $\lambda\in \N$ and all QPT $\adv_\lambda$ it holds that
%     \[
%     \Pr[k_0 \neq k_1 \text{ and } k_1 \neq \bot] = \negl(\lambda)
%     \]
%     where $k_0$ and $k_1$ are defined in the run of $\mathsf{QKDSec}^{\adv_\lambda}(1^\lambda)$.
% \end{definition}
%
Note that the above definition of verifiability is tight for two message protocols: An adversary can easily cause a disagreement between the keys by doing nothing on the first round, and blocking the second message. In this case $k_0$ would be a valid key (by correctness), whereas $k_1$ would be set to $\bot$, since the second message was never delivered.

\subsection{Two-Message QKD from QPKE}

We are now ready to describe our QKD protocol. Our ingredients are a QPKE scheme $(\SKGen, \allowbreak \PKGen, \allowbreak \Enc, \allowbreak \Dec)$ and a OTS scheme $(\SGen, \Sign, \Ver)$, which can be both constructed from one-way functions. Additionally, we will use a universal hash function family
\[
\mathsf{Hash}: \{0,1\}^{4\lambda} \to \{0,1\}^\lambda
\]
which exist unconditionally. For convenience, we denote by $s(\lambda)$ the size of a signature for a message of size $\lambda$. We present the protocol below.

\begin{itemize}
    \item $\QKDFirst(1^{\lambda})$:
    \begin{itemize}
        \item For all $i = 1,\dots, 4\lambda + s(4\lambda)$ sample a QPK key pair
        \[\sk_i \gets \SKGen(1^{\lambda})\quad \text{ and }\quad ({\rho}_i, \pk_i) \gets \PKGen(\sk_i).\]
        \item Set the first message to $\msg = (\pk_1, \dots, \pk_{4\lambda + s(4\lambda)})$ and $\mu = \rho_1 \otimes \dots \otimes \rho_{4\lambda + s(4\lambda)}$.
    \end{itemize}
    \item $\QKDSecond(\msg, \mu)$:
    \begin{itemize}
        \item Sample a OTS key pair $(\vk, \mathsf{zk}) \gets \SGen(1^\lambda)$, a key $k \gets \{0,1\}^{4\lambda}$, and a universal hash function $\mathsf{Hash}$.
        \item Compute $\sigma \gets \Sign(\mathsf{zk}, k)$.
        \item For all $i = 1,\dots, 4\lambda + s(4\lambda)$ compute
        \[
        \left\{\ct_i \gets \Enc({\rho}_i, \pk_i, k_i)\right\}_{i \leq 4\lambda} \quad\text{ and } \quad \left\{ \ct_{i} \gets \Enc({\rho}_{i}, \pk_{i}, \sigma_i)\right\}_{i> 4\lambda}
        \]
        where $k = (k_1, \dots, k_{4\lambda})$ and $\sigma = (\sigma_1, \dots, \sigma_{s(4\lambda)})$.
        \item If any of the encryption procedures fails, return $\bot$.
        \item Else, set the response as $\resp = (\mathsf{Hash}, \vk, \ct_1, \dots, \ct_{4\lambda + s(4\lambda)})$, no quantum state is present.
        \item Set the key to $K = \mathsf{Hash}(k)$.
    \end{itemize}
    \item $\QKDDecode(\st, \resp)$: 
    \begin{itemize}
        \item For all $i = 1,\dots, 4\lambda + s(4\lambda)$ compute
        \[
        \left\{k_i \gets \Dec(\sk_i, \ct_i) \right\}_{i\leq 4\lambda}\quad\text{ and }\quad \left\{\sigma_i \gets \Dec({\sk}_{\lambda + i}, \ct_{\lambda + i})\right\}_{i > 4\lambda}.
        \]
        \item If $\Ver(\vk, k, \sigma) \neq 1$ return $\bot$.
        \item Else, return $K = \mathsf{Hash}(k)$.
    \end{itemize}
\end{itemize}
Correctness follows immediately from the correctness of the underlying building blocks. Next we prove that the scheme is private and verifiable.
\begin{theorem}
    If quantum-secure one-way functions exist, then the QKD 
$(\QKDFirst, \QKDSecond, \allowbreak\QKDDecode)$ satisfies everlasting security (privacy and verifiability).
\end{theorem}
\begin{proof}
    We first show that the scheme satisfies privacy. Let us change the syntax of the experiment to explicitly include the view of the adversary the flag $\mathsf{abort}\in\{0,1\}$ that denotes whether the experiment aborted or not in step 5. Then, it suffices to show that the distribution of the keys $K_0$ (as defined step~3 of the experiment) and $K_1$ (as defined in step~5 of the experiment) are statistically close to uniform, conditioned on the variables
    \[
    \left\{ \mathsf{QKDSec}^{\adv_\lambda}(1^\lambda), \mathsf{abort} \right\}.
    \]
    As a first step, we claim that the min-entropy of $k$ is $H_\infty(k) \geq 3\lambda$ in the above view. To show this, we will first consider a modified distribution, where the distinguisher \emph{is not provided} the variable $\mathsf{abort}$. We then define $\mathsf{Hyb}_0^{\adv_\lambda}(1^\lambda)$ to be the output of the original experiment $\mathsf{QKDSec}^{\adv_\lambda}(1^\lambda)$ as defined in \cref{def:qkd_ever}. Then for $i = 1, \dots, 4\lambda + s(4\lambda)$ we define the hybrid $\mathsf{Hyb}_i^{\adv_\lambda}(1^\lambda)$ as follows.
    \begin{itemize}
        \item $\mathsf{Hyb}_i^{\adv_\lambda}(1^\lambda)$: This is defined as the previous hybrid, except that the $i$-th ciphertext $\ct_i$ is computed as
        \[
        \ct_i \gets \Enc({\rho}_i, \pk_i, 0).
        \]
    \end{itemize}
    The statistical indistinguishability of neighbouring outputs follows immediately from the everlasting security of the QPKE scheme, i.e., 
    \[
    \mathsf{Td}\left(\mathsf{Hyb}_{i-1}^{\adv_\lambda}(1^\lambda), \mathsf{Hyb}_{i}^{\adv_\lambda}(1^\lambda)\right) = \negl(\lambda) \text{ for all } i = 1, \dots, 4\lambda + s(4\lambda).
    \]
    Note that in the last hybrid $\mathsf{Hyb}_{4\lambda + s(4\lambda)}^{\adv_\lambda}(1^\lambda)$ the view of the adversary is formally independent from the key $k$ and therefore $k$ has exactly $4\lambda$ bits of entropy. An application of \cref{lemma:lhl} already shows that $K_0$ is statistically close to uniform, since its distribution is independent of the event $\mathsf{abort}$. 
    
    This is however not the case for $K_1$, since whether or not $K_1 = \bot$ depends on the event  $\mathsf{abort}$. Applying the same argument backwards, we lose only a negligible summand in the entropy of $k$, and so we can conclude that $H_\infty(k) \geq 4\lambda - 1$ in the original experiment. By \cref{eq:chain} (chain rule for average-case min entropy), we have that, conditioned on the event $\mathsf{abort}$, it holds that
    \[
    \Tilde{H}_\infty(k | \mathsf{abort}) \geq H_{\infty}(k) - 2 > 3\lambda
    \]
    since $\mathsf{abort}\in\{0,1\}$. By \cref{lemma:lhl}, the statistical distance of $K_1$ from uniform is bounded from above by $\varepsilon = 2^{-\lambda}$ since
    \[
    \Tilde{H}_\infty(k | \mathsf{abort}) - 2\log\left(\frac{1}{\varepsilon}\right) \geq 3\lambda - 2\lambda = \lambda = \ell.
    \]
    This shows that both $K_0$ and $K_1$ are statistically close to uniform, and concludes the proof of everlasting security.

    As for verifiability, let us assume towards contradiction that there exists a QPT adversary that causes a key mismatch $k_0 \neq k_1$ while not causing the decoding algorithm to reject, i.e., $k_1 \neq \bot$. Then it must be the case that the adversary is able to produce a valid message-signature pair $(k_1, \sigma^*)$ under $\vk$, for a message $k_1 \neq k_0$. Since the adversary runs in quantum polynomial time, this contradicts the unforgeability of the OTS scheme and concludes our proof. 
\end{proof}